\documentclass[11pt,a4paper]{article}

\usepackage{amsthm,amsmath,amsfonts,amssymb}
\usepackage{hyperref}
\usepackage{cleveref}
\usepackage{color,xcolor}

\usepackage{amssymb}
\usepackage{amsmath}
\usepackage{graphicx}
\usepackage{epsfig}
\usepackage{xspace}
\usepackage{pdfpages}
\usepackage{algorithm}
\usepackage{algorithmic}

\newcommand{\comment}[1]{}

\newcommand{\A}{{\mathcal A}}
\newcommand{\G}{{\mathcal G}}

\newcommand{\C}{{\cal C}}

\newcommand{\E}{{\cal E}}
\newcommand{\J}{{\cal J}}

\newcommand{\tuple}[1]{\langle #1  \rangle}

\newtheorem{theorem}{Theorem}[section]
\newtheorem{proposition}[theorem]{Proposition}

\newtheorem{claim}[theorem]{Claim}

\newtheorem{example}{Example}[section]

                      {}

\newcommand{\ceil}[1]{\left\lceil #1 \right\rceil}

\newcommand{\floor}[1]{\left\lfloor #1 \right\rfloor}

\begin{document}

\title{
Interval Scheduling Games with Color-Based Concurrent Jobs}

\author{Vipin Ravindran Vijayalakshmi\thanks{Viavi Solutions, Berlin, Germany, E-mail: vipin.rv@oms.rwth-aachen.de.} \and Marc Schr{\"o}der \thanks{School of Business and Economics, Maastricht University, Netherlands, E-mail: m.schroder@maastrichtuniversity.nl} \and
Tami Tamir\thanks{School of Computer Science, Reichman University, Israel. E-mail:tami@runi.ac.il.}}
\date{}
\maketitle

\begin{abstract}
  We consider a game-theoretic variant of an interval scheduling problem. Every job is associated with a length, a weight, and a color. Each player controls all the jobs of a specific color, and needs to decide on a processing interval for each of its jobs. Jobs of the same color can be processed simultaneously by the machine. A job is covered if the machine is configured to its color during its whole processing interval. The goal of the machine is to maximize the sum of weights of all covered jobs, and the goal of each player is to place its jobs such that the sum of weights of covered jobs from its color is maximized. The study of this game is motivated by several applications like antenna scheduling for wireless networks.
    
    We first show that given a strategy profile of the players, the machine scheduling problem can be solved in polynomial time. We then study the game from the players' point of view. We analyze the existence of Nash equilibria, its computation, and inefficiency. We distinguish between instances of the classical interval scheduling problem, in which every player controls a single job, and instances in which color sets may include multiple jobs.
\vspace{0.2cm}

    {\bf keywords:} Interval scheduling, Scheduling games,
Equilibrium inefficiency.
\end{abstract}

\maketitle

\section{Introduction}
Scheduling problems and game-theory are a fruitful and well studied combination. The machine scheduling problem that we consider in this paper is motivated by the beam selection problem on a base station \cite{kose2021profiling,li2020context}. Antennas in modern 5G base stations are designed to dynamically steer their beams to focus transmission signals towards specific users using a technique known as beam forming. Beam forming maximizes data throughput between the base station and the user equipment while minimizing the radio interference. Base stations serve multiple users by allocating time intervals during which it directs its beam toward each user. However, base station antennas often suffer the restriction that only a subset of beams can be activated at any given point in time \cite{kose2021profiling,li2020context}. The resource scheduler in a base station must choose these time intervals effectively to maximize overall network performance, especially in environments with multiple users competing for access to the radio resource, e.g., at a football stadium, by deciding when and for how long the radio resources are scheduled for a user. Beam forming techniques have become an integral part of 5G wireless communication and this gives rise to the optimization problem of deciding which beam is active at which point in time, having implications for the set of users that can be serviced.

We model this scheduling problem by means of a generalization of the classic interval scheduling problem \cite{kleinberg2006algorithm}. In the classic interval scheduling problem, we are given a machine and processing intervals for jobs so that the machine has to decide which jobs are rejected and which jobs are processed, subject to no two processed jobs having overlapping processing intervals. In our generalization, at each point in time the machine has to be configured to a certain color, a color can be thought of as a single beam direction in the beam selection problem, so that all jobs of this color can be serviced simultaneously. A job is then completed if the machine is configured to that color during all its processing interval. 




We are interested in a game-theoretic variant of the above interval scheduling problem. Assuming players know how the machine solves the beam selection problem, players compete for access to the radio resource. That is, players schedule their requests by selecting a time interval such that their requests can be processed. We assume that each player controls jobs of one color. The strategy of a player is to decide on a processing interval for each job of its color. The player's goal is to assign the jobs such that the machine will complete as many jobs of its color as possible. 

We consider three different problems related to the above description. First, we solve the machine scheduling problem that decides which jobs are covered given the processing interval of each job. Second, we study the problem of assigning jobs to processing intervals so that the sum of weights of all covered jobs is maximized. Third, we provide answers to basic  problems in the analysis of the corresponding game. Specifically, we study the existence and computation of a pure Nash equilibrium, computation and convergence of best-response dynamics, and the equilibrium inefficiency.
%

\section{Model}

An interval scheduling game with color-based concurrent jobs is given by the following tuple $\tuple{\J,\C, T, (p_j)_{j \in \J},(w_j)_{j \in \J}}$, containing a set of jobs $\J=\{1,\ldots,n\}$, a set of colors $\C$ with $|\C|=c$ and a time interval $[0,T)$, where every job $j\in \J$ has a color $c_j\in \C$, a length $0\leq p_j \le T$ and a weight $w_j\geq 0$.

 For $i\in \C$, denote by $J_i$ the jobs having color $i$. We assume that all the jobs of $J_i$ are controlled by one player. 
 A strategy, $\sigma_i$, of  player $i$ assigns a processing interval $[s_j,f_j)\subseteq [0,T)$ with $f_j-s_j=p_j$ to each job $j$ with $c_j=i$. A  profile $\sigma=(\sigma_i)_{i\in \C}$ assigns a strategy to each player. 

Given a strategy profile $\sigma$, the machine faces the following scheduling problem. We say that job $j$ is {\em covered} in a given schedule if the machine is configured to process color $c_j$ during $[s_j,f_j)$. Note that if the machine is configured to a specific color, the machine can service an unlimited number of jobs from that color simultaneously. Let $\A \subseteq \J$ be the set of covered jobs. The goal of the machine is to maximize $\sum_{j \in \A} w_j$. Formally, the output for the machine scheduling problem is a configuration for the machine, i.e., a partition of the interval $[0,T)$ to intervals $\{[0,t_1),\ldots,[t_{m-1},t_{m}=T)\}$ and a mapping $\gamma:\{1,\ldots,m\} \rightarrow \C$, such that for every $1 \le k \le m$, the machine is configured to process jobs of color $\gamma(k)$ during $[t_{k-1}, t_k)$.

Given strategy profile $\sigma$, the machine solves its optimization problem. Thus, determining for each job whether it is covered or not. The utility of each player $i\in \C$, denoted by $u_i(\sigma)$, is given by the sum of weights of covered jobs $j$ with $c_j=i$. Given $i\in \C$ and $\sigma_{-i}=(\sigma_j)_{j\neq i}$, we say a strategy $\sigma_i$ is a best-response for player $i$ if $u_i(\sigma_i,\sigma_{-i})\geq u_i(\sigma'_i,\sigma_{-i})$ for all $\sigma'_i$. A strategy profile $\sigma$ is called a \emph{Nash equilibrium} (NE) if for all $i\in \C$,
$u_i(\sigma_i,\sigma_{-i})\geq u_i(\sigma'_i,\sigma_{-i})\text{ for all } \sigma'_i.$

Best-response dynamics (BRD) is a local-search method where in each iteration some player is chosen and plays its best strategy given the strategies of the other players. When applied to our game, every iteration consists of two steps. First, an arbitrary player $i\in \C$ is chosen and may modify the location of the intervals of $J_i$. Then, as a response, the machine may modify its configuration along $[0,T)$. The dynamics then proceed to the next iteration, until no player has a beneficial deviation. Note that BRD need not always stop. 

\begin{example}\label{ex:inexis}
Consider the following game with $T=4$ and $c=2$. Player $1$ controls the set $J_1$, which includes two jobs, where $p_1=4$ and $p_2=1$, both having weight $2$. Player $2$ controls the set $J_2$, which includes a single job of length $p_3=1$ and weight $3$.

\begin{figure}[ht]
\centering
\includegraphics[width=0.7\textwidth]{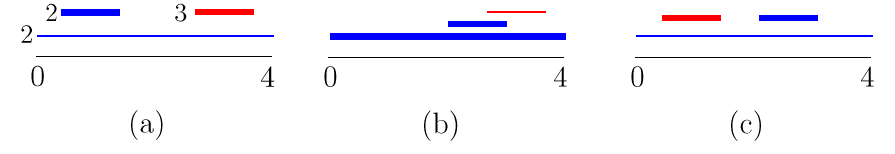}
\caption{\normalfont{
A game $G$ with $c=2$ that has no NE. Intervals are labeled by their weights. Bold intervals are covered by the machine.}}
\label{fig:NoNE2}
\end{figure}

If the two unit jobs are placed in disjoint intervals (See Fig.~\ref{fig:NoNE2}(a)) then the machine will process them. The utilities of the players are $(2,3)$. If the two unit jobs are placed in overlapping intervals (Fig.~\ref{fig:NoNE2}(b)) then the machine will cover only jobs of $J_1$. The utilities of the players are $(4,0)$. Therefore, the game has no NE: player $2$ will place its unit job such that it does not overlap with the unit-job of player $1$, and player $1$ will place its unit-job such that it overlaps with job $3$. Since $T=4$,  player $2$ always has a valid move (Fig.~\ref{fig:NoNE2}(c)). 
\end{example}

Denote the set of Nash equilibria of a given game $G$ by $\E$. Define the profit of strategy profile $\sigma$ by $\mathrm{val}(\sigma)=\sum_{i\in C}u_i(\sigma)$.
We identify classes for which an NE is guaranteed to exist, provide algorithms for computing an NE, study the convergence of best-response dynamics, and analyze the inefficiency of Nash equilibria.  To this end, a strategy profile $\sigma$ that maximizes $\mathrm{val}(\sigma)$ is denoted a social optimum solution (OPT) and  its profit is denoted by $\mathrm{val}(\mathrm{OPT})$. The \emph{price of anarchy} ($\mathrm{PoA}$) and the \emph{price of stability} $(\mathrm{PoS})$ of a game $G$ are defined as follows.
$$\mathrm{PoA}(G)=\frac{\mathrm{val}(\mathrm{OPT})}{\min_{\sigma\in \E}\mathrm{val}(\sigma)}~~ {\mbox and} ~~ 
 \mathrm{PoS}(G)=\frac{\mathrm{val}(\mathrm{OPT})}{\max_{\sigma\in \E}\mathrm{val}(\sigma)}.$$


Some of our results refer to restricted classes of games, specifically
\begin{enumerate}
  \item $\G_{single}$ - games corresponding to the classic interval scheduling problem -- in which there is exactly one job from each color, that is, $n=c$. 
    \item $\G_{unit}$ - games with unit-length jobs, that is, for every job $j \in \J$ it holds that $p_j=1$. 
    \item $\G_{prop}$ - games in which the weight of a job is proportional to its length, that is, for every job $j \in \J$ it holds that $w_j=p_j$. 
\end{enumerate}

For a class of games, $\mathcal{G}$, 
the {\em price of anarchy} of $\G$ is 
$\mbox{PoA}(\mathcal{G}) = sup_{ G\in \mathcal{G}}\mbox{PoA}(G)$,
and the {\em price of stability} of $\G$ is 
$\mbox{PoS}(\mathcal{G}) = sup_{ G\in \mathcal{G}}\mbox{PoS}(G)$.

\subsection{Our Results}
Our first result, Theorem \ref{thm:mach}, shows that given a strategy profile of the players, the machine scheduling problem can be solved in polynomial time using dynamic programming.

We then proceed to show (Theorem \ref{th:SOconst}), that computing a socially optimal solution for the players can be done in polynomial time if the number of colors is constant, but (Theorem \ref{th:SOhard}) it is weakly NP-hard when the number of colors is part of the input. We then present (Theorem \ref{th:SOPP}) a pseudo-polynomial algorithm for this task whose time complexity is $O(ncW)$ for $W=\sum_i w_i$ and integral weights. 

We next analyze Nash equilibria of our game. We show in Example \ref{ex:inexis} that in general, even with two players, a (pure) Nash equilibrium need not exist. However, for the game-theoretic variant of the classic interval scheduling problem, where every player controls a single job, that is, the class $\G_{single}$, we prove in Theorem \ref{th:exist} that a (pure) Nash equilibrium always exists and the price of stability is $1$. For this class of games, we present in Theorem \ref{thm:NE_single_exists} an algorithm for computing an NE. Then, Theorem \ref{th:poa} provides a complete and tight analysis of the inefficiency of Nash equilibria. For $n\le 5$, we show that the price of anarchy is at most $2$, and for $n \ge 6$, we show that the price of anarchy is at most $(n-1)/2$. Moreover, these bounds are tight. We provide a significantly lower upper bound of $3$ in Theorem \ref{th:w=p} for the restricted setting of the classic interval scheduling problems with processing times proportional to the weights.

For general games, we prove in Theorem \ref{thm:SNPHdecide} and Theorem \ref{thm:BRD2Shard} that even with only two players, computing a best-response of a player, as well as deciding whether an instance has a Nash equilibrium are both strongly NP-hard. In Theorem \ref{thm:poapos} we present an upper bound of $c$ for the price of anarchy, and show that even a best NE can have poor quality, by proving that the price of stability of this class is $2$. 

A second class of interval scheduling games with a guaranteed Nash equilibrium are games with unit processing times. For this class, the price of anarchy is significantly lower. Specifically, in Theorem \ref{thm:unitPoA} we provide a tight bound of $\min\{3-2/c, 3- 2/\floor{T}\}$.

Lastly, we define a natural extension of the game, in which jobs may be associated with different release times and due dates, and show in Theorem \ref{thm:noNEnonsymm} that even instances with unit processing time might lack equilibria.

\subsection{Related Literature}
The machine scheduling problem we study is a generalization of the classical interval scheduling/activity selection problem. This problem can be solved greedily for unweighted jobs \cite{carlisle1995k,faigle1995note}, and by dynamic programming \cite{kleinberg2006algorithm} or min-cost flow computation \cite{arkin1987scheduling,bouzina1996interval} for weighted jobs. For an overview of different variants of interval scheduling, we refer to \cite{kolen2007interval,kovalyov2007fixed}. Our analysis shows a close connection to the famous knapsack problem --- one of Karp's NP-complete problems \cite{Kar72}.

There are many game-theoretic models of job scheduling problems. The majority of this literature refer to games in which each
job selfishly chooses a machine so as to minimize a certain objective. Different authors make various assumptions about the machines' scheduling policy, the cost function of the players, and the social welfare function to be considered, like the makespan of the schedule or the sum of completion times. This line of research was initiated by \cite{Koutsoupias:1999:WE:1764891.1764944} and later extended by, for example, \cite{bilo2020congestion,cole2015decentralized,correa2012efficiency,czumaj2007tight,durr2009non,vijayalakshmi2021scheduling}.

{\em Real-time scheduling}, refers to a scheduling environment in which  jobs are associated with intervals during which they have to be processed. There is wide literature on real-time scheduling, either on a single or on parallel machines (see surveys in \cite{CGK14,IP05}). Most of the existing work consider systems controlled by an external authority determining the jobs' assignment. When the server has a limited capacity, and jobs have variable-weights, many problems such as minimizing the number of late jobs, or minimizing the servers' busy time are NP-hard, even with unit-length jobs \cite{CE+11,Alb10}. On the other hand, with unit-weight unit-length jobs, these problems are polynomially solvable \cite{Bap00,CGK14}.  
Real-time scheduling has been studied as a cost-sharing game, in which jobs that are processed in the same interval, share the machine's activation cost in this interval~\cite{Tamir22,GKT21}. In the above games, the machine may process multiple jobs simultaneously and the congestion on the machine during its processed interval determines the job's cost. The paper \cite{herzel2019multistage} considers a strategic variant of the multistage interval scheduling problem, in which jobs consist of several tasks, and all tasks have to be scheduled in order to obtain the profit associated with the job.

Another related game is \emph{selfish bin packing}. The goal in this line of research is to find a cost-sharing rule so that number of bins used by selfish jobs is as close as possible to the optimal number of bins \cite{bilo2006packing,wang2022best,epstein2011selfish,yu2008bin}.
A different setting in which players correspond to jobs who need to be processed by a machine with a limited capacity arises in {\em knapsack auctions}, where each job suggests a payment for being processed \cite{MN08,BKV11}.  We note that in the above packing games, a profile does not specify the location of a packed item in the knapsack, unlike our game, in which the specific interval in which a job is processed within the interval $[0,T)$ plays a crucial role.

Interval scheduling games with single job classes are closely related to \emph{claim games} for the estate division problem (see \cite{o1982problem,atlamaz2011non,peters2019claim}). In claim games, players claim parts of the estate by splitting their claim over the estate and subintervals that are claimed by multiple players are then divided among those players. Our scheduling game differs from that setting in that each job can only be scheduled consecutively and jobs have to be completely covered.

\section{Optimal Algorithm for the Machine Scheduling Problem}\label{sec:machine}

In this section we show that we can compute an optimal configuration for the machine in polynomial time. 
The input for the problem is a placement of the jobs in $[0,T)$, and the output for the machine scheduling problem is a configuration for the machine. The objective is to maximize the weight of covered jobs.

Observe first that if $c=n$, that is, every player controls a single job, then the machine's scheduling problem is the classical interval scheduling/activity selection problem~\cite{kleinberg2006algorithm}.

We present an optimal solution for the most general instance of arbitrary length, arbitrary weight jobs, and arbitrary number of colors. The following notations are used in our solution. We assume that the jobs are indexed by finish time, that is, $f_1 \le \ldots \le f_n$.

\begin{itemize}
\item
For a set of jobs $J \subseteq \J$, let $w(J):= \sum_{j \in J} w_j$ be the total weight of jobs in $J$.
 \item
For every job $j \in \J$, let $prev(j):=\max\{k\in\J\mid f_k < s_j\}$, that is, $prev(j)$ is the index of the last job to end before job $j$.
If there is no such job, then define $prev(j)=0$.
\item
For every job $j \in \J$, let $In(j):=\{k\in\J\mid [s_k,f_k) \subseteq [s_j,f_j) \mbox{ and } c_k=c_j\}$, that is, $In(j)$ is the set of jobs that have the same color as job $j$ and whose interval is fully included in the interval of job $j$. In particular, $j \in In(j)$.
\end{itemize}


\begin{theorem}\label{thm:mach}
 There is a dynamic program that, for each strategy profile $\sigma$, solves the machine scheduling problem for $\sigma$ in polynomial time. 
\end{theorem}
\begin{proof}
 Define $A[j]$ to be the maximum profit from a solution in which job $j$ is the last covered job. In other words, $A[j]$ represents the maximum profit from an instance that only includes jobs $1,\ldots,j$ when job $j$ is covered.

The base case is $A[0]=0$. 
The recursive  formula for $j\geq 1$ is 
$A[j]= \max\{X, Y\}$, where
    $$X = \max_{k \le prev(j)} A[k] + w(In(j))$$ and
$$Y=  \max_ {k<j \mbox{ and }c_k = c_j \mbox{ and } k \not \in In(j) }  A[k]  + w(In(j) \setminus In(k)).$$

The value of $X$ is the maximum profit in case the last job covered before job $j$ ends before job $j$ starts. The value of $Y$ is the maximum profit in case job $j$ starts while the machine services earlier jobs of color $c_j$. In this case, we select the intersecting job from $j$'s color, for which the added profit from extending the interval in which the machine is configured to $c_j$ is maximal. Note that the case in which the machine is configured to $c_j$ and is then idle, is covered in both $X$ and $Y$.

Once the table $A$ is full, the value of an optimal solution is given by $\max_j A[j]$. The optimal schedule itself can be retrieved by backtracking. We assume that in case of multiple optimal solutions, the machine selects the subset of jobs that is lexicographically smallest.

The solution involves a pre-processing in which $prev(j)$, $In(j)$, $w(In(j))$ are computed for all $j$, as well as a computation of $In(j) \setminus In(k)$ and $w(In(j) \setminus In(k))$ for all $k<j$ and $c_k=c_j$. Since the jobs are sorted by finish time, all these values can be computed in total time $O(n^2\log n)$ using interval trees. Once these values are computed, computing $A[j]$ takes $O(n)$-time for a total of $O(n^2)$. Thus, the total time complexity of the algorithm is $O(n^2 \log n)$.
\end{proof}


\section{Computing a Social Optimum of a Game}
In this section we consider the optimization problem of finding a strategy profile that maximizes the sum of weights of covered jobs. Note that this task should not be confused with the problem of computing an optimal schedule for the machine (discussed in Section~\ref{sec:machine}). The problem considered below is the following: Given an instance $\tuple{\J,\C, T, (p_j)_{j \in \J},(w_j)_{j \in \J}}$ of the game, place the jobs in $[0,T)$, such that an optimal configuration of the machine produces maximal profit.
We first prove the following property.
\begin{claim}
\label{cl:single}
For all $G$, there exists a socially optimal solution in which the machine processes jobs of each color $i\in \C$ during at most one interval. 
\end{claim}
\begin{proof}
Assume that some socially optimal solution has a different structure. That is, for some color $i$, the machine processes color $i$ in two disjoint intervals.
Specifically, let $[a, b), [b, c)$, and $[c,d)$ be adjacent intervals such that $a < b <c$ and the machine processes jobs from $J_i$ during both $[a, b)$ and $[c,d)$, and jobs from other classes during $[b,c)$. Clearly, the profit from every interval $[x,y)$ is achieved only due to jobs of length at most $y-x$. Consider a schedule in which the jobs covered during $[c,d)$ are shifted to be covered during $[b,b+d-c)$ and the jobs covered during $[b,c)$ are shifted to be covered in $[b+d-c,d)$.
Note that the machine processes jobs of $J_i$ during $[a,b+d-c)$. It is easy to see that the profit of the machine does not reduce, since every job that was covered in the original schedule is covered also in the modified schedule. Therefore, as long as there is a color covered in more than one interval, it is possible to modify the schedule and reduce the number of intervals, without decreasing the sum of weights of covered jobs.     
\end{proof}

Based on the above claim, we show the following: 

\begin{theorem}
\label{th:SOconst}
If $c$ is constant, then a socially optimal solution can be computed in polynomial time.
\end{theorem}
\begin{proof}
    For simplicity, add a dummy job of length $0$ and weight $0$ to each set $J_i$. Let $n_i=|J_i|$.
By Claim~\ref{cl:single}, the machine processes jobs of color $i$ in a single interval. We can assume w.l.o.g., that jobs of $J_i$ are covered in interval $[t_{i-1},t_i)$, such that $t_0=0, t_c=T$ and $\sum_{i=1}^c (t_i-t_{i-1})=T$.  Also w.l.o.g., this interval allocated to color $i$ has length $p_j$ for some $j \in J_i$. Therefore, a socially optimal solution can be calculated by considering at most $\Pi_{i=1}^c n_i$ candidate solutions. Specifically, for every $(a_1,a_2,\ldots,a_c)$, where for all $i$, $0 \le a_i \le n_i$, let $p_{a_i}(i)$ be the length of the $a_i$-th shortest job in $J_i$ and let $W_{a_i}(i)$ be the total weight of the shortest $a_i$ jobs of $J_i$. The solution induced by $(a_1,a_2,\ldots,a_c)$ is feasible if $\sum_{i=1}^c p_{a_i}(i) \le T$ and its profit is $\sum_{i=1}^c W_{a_i}(i)$. 
\end{proof}

\begin{theorem}
\label{th:SOhard}
Computing a socially optimal solution is weakly NP-hard, already for the class $\G_{single}$.
\end{theorem}
\begin{proof}
The proof is by a simple reduction from the $0-1$ Knapsack problem. Note that in our reduction $|J_i|=1$ for all $1 \le i \le c$.

An instance of the knapsack problem is given by a set of $c$ items each associated with a weight $p_i$ and a value $w_i$. A subset of these items should be placed in a knapsack with capacity $T$ such that the total value of the packed items is maximal. 

Given an instance of Knapsack, consider a game played in the interval $[0,T)$. There are $c$
colors, where $J_i$ consists of a single job of length $p_i$ and weight $w_i$. It is easy to see that there is a packing with total value $W$ if and only if there is a valid schedule with total profit $W$.
\end{proof}

Next, we show that the above hardness result is tight, in a sense that computing a social optimum can be done in pseudo-polynomial time if job weights are integral.
\begin{theorem}
\label{th:SOPP}
Computing a socially optimal solution can be done in time $O(ncW)$ for $W=\sum_i w_i$ if weights are integral. 
\end{theorem}

\begin{proof}
We present an optimal algorithm, based on dynamic programming (DP). 
The DP table, $MinTime$, is of size $(c+1)(W+1)$, where for every $0 \le i \le c$ and $0 \le w \le W$, $MinTime[i,w]$ is the minimal time required to achieve profit $w$ from the first $i$ color sets. If it is impossible to achieve profit $w$ from the first $i$ color sets, then  $MinTime[i,w]=\infty$.

By Claim~\ref{cl:single} there exists a socially optimal solution in which the machine processes jobs of each color $i\in \C$ during at most one interval. Since all jobs are available in $[0,T]$, the internal order of these intervals is flexible, and in particular, we can assume that jobs of $J_1$ are processed first, then jobs of $J_2$, etc. 

For simplicity, add a dummy job of length $0$ and weight $0$ to each set $J_i$. The dummy job is indexed $0$. Let $n_i=|J_i|$.
For every $0 \le \ell \le n_i$, let $p^i_{\ell}$ be the length of the $\ell$-th job in $J_i$ and let $w^i_{\ell}$ denote the total weight of jobs of length at most $p^i_{\ell}$ in $J_i$. 
As a base case, we set $MinTime[0,0]=0$ and $MinTime[i,w]=\infty$ for every other entry as well as for $w<0$. The general formula is

$$MinTime[i,w] = \min_{0 \le \ell \le n_i} MinTime[i-1, w-w^i_{\ell}] +p^i_{\ell}.$$

Since there is an optimal solution in which for every $1 \le i \le c$, $J_i$ is allocated an interval of length $p^i_{\ell}$, for some $0 \le \ell \le n_i$, the above formula considers, among all, also the optimal assignment. 

The calculation of $p^i, w^i$ is done in $O(n log n)$ as a preprocessing. 
We can then fill the table for every $1 \le i \le c$ and $1 \le w \le W$. The computation of a single entry takes time $O(n_i)= O(n)$, for a total of $O(ncW)$. Once the table is full, we return the  maximal $w$ for which $Time[c,w] \le T$.
\end{proof}

We note that our dynamic programming algorithm assumes integral job weights. This restriction aligns with the definition of a pseudo-polynomial time algorithm, whose complexity is polynomial in the unary length of the input. 


\section{One Job per Color} 
In this section we analyze the class $\G_{single}$ where $c=n$. This class corresponds to the classic interval scheduling problem. For every color $1 \le i \le c$, denote by $p_i, w_i$ the length and the weight of the single job in $J_i$.

Every game $G \in \G_{single}$ induces a $0-1$ Knapsack problem, with a knapsack of capacity $T$, and $n$ items, where item $i$ has size $p_i$ and value $w_i$. Some of the results in this section leverage the relation between the game and its corresponding packing problem. We note the following crucial differences between the problems. First, in the knapsack problem, the physical location of the items is not part of the solution, whereas in our setting, it plays a pivotal role. Second, in the knapsack problem, items that are not packed are simply rejected, whereas in our game, every job must be placed somewhere in the interval 
$[0,T)$; the machine selects a set of non-overlapping jobs, corresponding to the packed items.

\begin{theorem}\label{th:exist}
Every game $G \in \G_{single}$ has an NE profile, and $\mathrm{PoS}(\G_{single})=1$.   
\end{theorem}
\begin{proof}
Observe that an optimal solution of the corresponding knapsack problem induces an NE, by placing the corresponding jobs one after the other along the interval $[0,T)$. This placement is an NE since no non-covered job can be added without rejecting jobs of at least the same weight. This relation with the knapsack problem also implies that $\mathrm{PoS}(G)=1$ for all $G \in \G_{single}$. 
\end{proof}

\begin{proposition}
\label{prop:brd}
For every game $G \in \G_{single}$, best-response dynamics converges to an NE.
\end{proposition}

\begin{proof}
A deviation is beneficial for a job $j \in \J$ only if $j$ is not covered before the deviation and becomes covered afterward. Consider the machine's profit. After job $j$'s deviation, the machine can maintain its configuration over the interval $[0,T)$ and thus its profit. Therefore, the machine modifies its strategy to cover $j$, only if this increases its profit, or, based on the machine's tie-breaking rule, a lexicographically smaller subset of jobs of the same profit is covered. Since the maximum profit from covered jobs is bounded and there are finitely many subsets of jobs, the result follows.
\end{proof}

We turn to consider the problem of computing an NE. 
It might seem that any approximation algorithm for the knapsack problem could easily be adapted to our game by scheduling,one after the other, the jobs corresponding to the packed items. However, the resulting schedule is not necessarily an NE. While best-response dynamics are guaranteed to converge to an NE with at least the same profit, there does not seem to be a  straightforward way to bound the convergence time. However, as we show, an NE can be computed efficiently.

\begin{theorem}
\label{thm:NE_single_exists}
For every game $G \in \G_{single}$, computing an NE profile can be done in polynomial time.
\end{theorem}
\begin{proof}
We present an algorithm for computing an NE. Recall that for every job $j \in \J$, it holds that $0 \le p_j \le T$. We also assume that for at least one pair of jobs, $x,y$ we have that $p_x+p_y \le T$, as otherwise, covering just the most profitable job is clearly an NE.

\begin{algorithm}[H]
\caption{ - Computing an NE schedule of a game $G \in \G_{single}$}
\label{alg:single}
\begin{algorithmic}[1]
\STATE Let $h$ be a job with maximal weight.
\STATE Let $\{a,b\}$ be a pair of jobs for which $w_a+w_b$ is maximal among all pairs $x,y$ fulfilling $p_x+p_y \le T$. Possibly, $a=h$ or $b=h$.
\IF {$w_a+w_b \le w_h$}
\STATE For every $i \in \J$, schedule job $i$ in $[0,p_i)$ and halt
\ELSE
\STATE Schedule job $a$ in $[0, p_a$) and job $b$ in $[p_a, p_a+p_b)$.
\STATE Sort $\J \setminus\{a,b\}$ such that $w_1 \ge w_2 \ge \ldots\ge w_{n-2}$.
\STATE Let $C=p_a+p_b$.
\FOR{$i=1$ to $n-2$}
\IF{$C+p_i \le T$}
\STATE Schedule job $i$ in $[C, C+p_i$).
\STATE $C=C+p_i$.
\ELSE
\STATE Schedule job $i$ such that it overlaps both $a$ and $b$.
\ENDIF
\ENDFOR
\ENDIF
\end{algorithmic}
\end{algorithm}

The algorithm distinguishes between two cases. In the first case, the most profitable job, $h$, is more profitable than any pair of jobs that can fit together in $[0,T)$. This case is handled in lines 3-4. In the second case, let $a,b$ be a most profitable pair that can fit in $[0,T)$ and have total profit higher than $h$. 
This case is handled in lines 6-16.

Assume first that the condition in line 3 is valid. Thus, for any job $i \neq h$, $p_h +p_i >T$. Indeed, if for some job $i \neq h$, $p_h +p_i \le T$, then $w_a+w_b\geq w_h+w_i>w_h$, contradicting the condition in line 3.
This implies that the schedule produced in step 4 is an NE: 
all jobs overlap at $t=0$, the machine processes job $h$, and no job $a$ can move to an interval disjoint with $h$ or disjoint with a job $b$ such that $w_a + w_{b} > w_h$. 

Assume next that the condition in line 3 is not valid, that is, $w_a+w_b > w_h$. The algorithm first schedules $a$ and $b$ one after the other in the leftmost position of the schedule, and then considers the remaining jobs in non-increasing order by weight. A job that fits is placed in the leftmost idle position; a job that does not fit, is placed such that it overlaps both $a$ and $b$. Let $\sigma$ denote the resulting schedule. 

Since $w_a + w_b > w_h$, the machine will cover $a$ and $b$ and the jobs that were placed one after the other during the loop in lines 7-16. The jobs that intersect $a$ and $b$ will not be covered. We show that $\sigma$ is an NE. 

For every non-covered job $i$, at the time $i$ is considered by the algorithm, the total busy time of the machine is more than $T-p_i$. Assume that $i$ moves to an interval that begins after $p_a+p_b$. By the algorithm, such an interval must intersect at least one job that is more profitable than $i$, and therefore, covering $i$ is not beneficial for the machine. Assume next that $i$ moves to an interval that intersects $[0, p_a+p_b]$. Note that at most one additional job that intersects $[0, p_a+p_b]$ (possibly $a$ or $b$) can be covered together with job $i$. However, if covering a different pair that begins in $[0,p_a+p_b)$ is beneficial for the machine, we get a contradiction to the choice of $\{a,b\}$ in step 2 of the algorithm.
\end{proof}

We turn to analyze the equilibrium inefficiency. While the price of stability is $1$, we show that the price of anarchy is linear in $n$, which means that highly inefficient equilibria exist as well.

\begin{theorem}\label{th:poa}
Let $G \in \G_{single}$. If $n \le 5$ then $\mathrm{PoA}(G) \le 2$. If $n \ge 6$, then $\mathrm{PoA}(G) \le \frac{n-1}{2}$.
\end{theorem}
\begin{proof}
Let $\mathrm{OPT}$ be an optimal profile. Denote by $n^* \le n$ the number of jobs that are covered in OPT. Sort the jobs covered by OPT such that $w_1 \ge w_2 \ge \ldots, \ge w_{n^*}$. 

Since job $1$ can be selected by the machine, for any NE profile $\sigma$ it holds that $\mathrm{val}(\sigma) \ge w_1$.
If $n^*=2$ then  $\mathrm{PoA}(G) \le 2$, since $\mathrm{OPT} = w_1+w_2$ and $w_1\geq w_2$.

Consider next the case $n^*=3$. If $w_1 \ge \mathrm{OPT}/2$, then clearly, $\mathrm{PoA}(G) \le 2$. If $w_1 < \mathrm{OPT}/2$ then any two jobs in $A_3=\{1, 2,3\}$ have total profit at least $\mathrm{OPT}/2$. Consider an NE $\sigma$. If at least two jobs from $A_3$ are covered, then $\mathrm{val}(\sigma) \ge \mathrm{OPT}/2$. Otherwise, assume that job $1$ is assigned to interval $[t,t+p_1)$. 
Note that at least one of the intervals $[0,t)$, and $[t+p_1, T)$ have length at least $\frac{T-p_1}{2}$. Since all the jobs in $A_3$ are covered in $\mathrm{OPT}$, we have that $p_2+p_3 \le T-p_1$, thus $\min\{p_2,p_3\} \le \frac{T-p_1}{2}$. 
If $\mathrm{val}(\sigma) < \mathrm{OPT}/2$, each of job $2$ and $3$ intersects job $1$. However, in this case, at least one of $2$ and $3$ fits into $[0,t)$ or $[t+p_1, T)$, implying that at least one of these two jobs has a beneficial migration. Therefore, if $n^*=3$ then $\mathrm{PoA}(G) \le 2$.

We turn to consider the case $n^* \ge 4$.
Our proof distinguishes between the case $n^*=n$ and the case $n^* \le n-1$. 
 
{\bf Case 1: $n^* = n$:} Assume that all jobs are covered in $\mathrm{OPT}$. Note that this implies that $\sum_{i=1}^n p_i \le T$. 
Let $\sigma$ be an NE profile. Denote by $\J_{in}(\sigma)$ and $\J_{out}(\sigma)$ the sets of jobs that are covered and non-covered, respectively, by the machine in $\sigma$. Let $n_{in}(\sigma) = |\J_{in}(\sigma)|$ and $n_{out}(\sigma) = |\J_{out}(\sigma)|$. Note that $n_{out}(\sigma)+n_{in}(\sigma)=n$.
\begin{claim}
\label{cl:halfIN}
$n_{out}(\sigma) \le n_{in}(\sigma)$.
\end{claim}
\begin{proof}
    Let $B=\sum_{i\in \J_{in}(\sigma)} p_i$ be the total busy time of the machine in $\sigma$. Since all jobs in the instance can fit into $[0,T)$, the machine is idle at least $\sum_{i\in \J_{out}(\sigma)}p_i$ time units. Thus, $\sum_{i\in \J_{out}(\sigma)} p_i\le T-B$. There are $n_{in}(\sigma)$ covered jobs, hence, at most $n_{in}(\sigma)+1$ idle intervals, where at least one of the idle intervals has length at least $\frac{T-B}{n_{in}(\sigma)+1}$.     

By averaging, the shortest non-covered job has length 
at most $\frac{\sum_{i\in \J_{out}(\sigma)}p_i}{n_{out}(\sigma)} \le \frac{T-B}{n_{out}(\sigma)}$. Given that $\sigma$ is an NE, no non-covered job can fit into an idle interval, therefore, $\frac{T-B}{n_{out}(\sigma)} > \frac{T-B}{n_{in}(\sigma)+1}$, implying $n_{out}(\sigma) < n_{in}(\sigma)+1$, therefore, $n_{out}(\sigma) \le n_{in}(\sigma)$.
\end{proof}

\begin{claim}
\label{cl:pairs}
For any two jobs $a, b \in \J_{out}(\sigma)$, it holds that $w_a+w_b \le \mathrm{val}(\sigma)$.
\end{claim}
\begin{proof}
Let $x$ be the leftmost job in $\J_{in}(\sigma)$ such that $f_x \ge p_a$. Let $X$ be the set of jobs in $\J_{in}(\sigma)$ that are covered in $[0,f_x)$. For an illustration, see Fig. \ref{fig:ab}. Since $\sigma$ is an NE, $w_a \le \mathrm{val}(X)$.
Similarly, let $y$ be the rightmost job in $\J_{in}(\sigma)$ such that $s_y \le T-p_b$. Let $Y$ be the set of jobs in $\J_{in}(\sigma)$ that are covered in $[s_y,T)$. Since $\sigma$ is an NE, $w_b \le \mathrm{val}(Y)$. Also, $X \cap Y = \emptyset$ since all jobs can fit into $[0,T)$ and in particular $p_a+p_b+p_x+p_y \le T$. We conclude that $w_a+w_b \le \mathrm{val}(X \cup Y) \le \mathrm{val}(\sigma)$. 
\end{proof}

\begin{figure}[ht]
\centering
\includegraphics[width=0.4\textwidth]{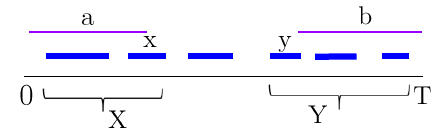}
\caption{\normalfont{
The non-covered jobs $a$, $b$, and the sets $X,Y$. The jobs of $\J_{in}(\sigma)$ are bold.}}
\label{fig:ab}
\end{figure}
We use the above claims to conclude the analysis for $n^*=n \ge 4$ as follows.
\begin{claim}
\label{cl:cge4}
If $n=n^*=4$, then $\mathrm{PoA}(G) \le 2$, and if $n=n^*\ge 5$, then $\mathrm{PoA}(G) \le \frac{n-1}{2}$.
\end{claim}

\begin{proof}
Recall that  $\mathrm{OPT} = \mathrm{val}(J_{in}(\sigma))+ \mathrm{val}(J_{out}(\sigma))= \mathrm{val}(\sigma)+ \mathrm{val}(J_{out}(\sigma))$. Let $n=4k+r$ for $0 \le r \le 3$.
\begin{itemize}
\item If $r=0$, then  by Claim~\ref{cl:halfIN}, $n_{out}(\sigma)\le 2k$. Claim~\ref{cl:pairs} implies that $\mathrm{val}(\J_{out}(\sigma)) \le k\cdot \mathrm{val}(\sigma)$.
We get that $\mathrm{PoA}(G) \le 1+k = 1 + \frac n 4$. For all $n \ge 6$ it holds that $1+ \frac n 4 \le \frac{n-1}2$. Note that for $n=4$ we have a bound of $\mathrm{PoA}(G) \le 2$, thus, this case covers correctly all $n=4k$ for $k \ge 2$.

\item If $r=1$, then  by Claim~\ref{cl:halfIN}, $n_{out}(\sigma)\le 2k$. Claim~\ref{cl:pairs} implies that $\mathrm{val}(\J_{out}(\sigma)) \le k\cdot \mathrm{val}(\sigma)$. We get that $\mathrm{PoA}(G) \le 1+k = 1 + \frac {n-1} 4$. For any $n \ge 5$ it holds that $1+ \frac {n-1} 4 \le \frac{n-1}2$. 
\item If $r=2$, then  by Claim~\ref{cl:halfIN}, $n_{out}(\sigma)\le 2k+1$. Claim~\ref{cl:pairs} implies that $\mathrm{val}(\J_{out}(\sigma)) \le (k+\frac 1 2)\cdot \mathrm{val}(\sigma)$.
We get that $\mathrm{PoA}(G) \le 1+(k+\frac 1 2) = \frac 3 2 + \frac {n-2} 4$. For any $n \ge 6$ it holds that $\frac 3 2+ \frac {n-2} 4 \le \frac{n-1}2$.
\item If $r=3$, then  by Claim~\ref{cl:halfIN}, $n_{out}(\sigma)\le 2k+1$. Claim~\ref{cl:pairs} implies that $\mathrm{val}(\J_{out}(\sigma)) \le (k+\frac 1 2)\cdot \mathrm{val}(\sigma)$. We get that $\mathrm{PoA}(G) \le 1+ (k+\frac 1 2) = \frac 3 2 + \frac {n-3} 4$. For any $n \ge 5$ it holds that $\frac{3}{2}+ \frac {n-3} 4 \le \frac{n-1}2$. 
\end{itemize}
\end{proof}

{\bf Case 2: $n^* \le n-1$:} If $w_1 > \frac{2\mathrm{OPT}}{n^*}$, then we are done. 
Otherwise, by averaging, for all $1 \le i \le n^*$, we have that $\sum_{k=1}^i w_i \ge \frac {i \cdot \mathrm{OPT}}{n^*}$. We show that job $1$, together with any of $2$ or $3$  are sufficient to achieve profit $\frac {2 \cdot \mathrm{OPT}}{n^*}$. Formally:
\begin{claim}
\label{cl:J1andJ23}
    For $i=2,3$ it holds that $w_1+w_i \ge \frac {2 \cdot \mathrm{OPT}}{n^*}$.
\end{claim}
\begin{proof}
    By averaging, $w_1+w_2 \ge \frac {2 \cdot \mathrm{OPT}}{n^*}$.
The proof for $w_1+w_3$ is a bit more involved. Let $w_1=\frac{\mathrm{OPT}}{n^*}+x+z$, $w_2=\frac{\mathrm{OPT}}{n^*}+x$ for some $x,z$ with $x+z\geq0$ and $z\geq 0$. Assume by contradiction that $w_1+w_3 < \frac{2\mathrm{OPT}}{n^*}$. Therefore, $w_3=\frac{\mathrm{OPT}}{n^*}-y$ for some $y > x+z\geq 0$. The total weight of the remaining $n^*-3$ jobs is $\mathrm{OPT} - \frac{3\mathrm{OPT}}{n^*} -2x-z+y= \frac{(n^*-3)\mathrm{OPT}}{n^*} -x - (x+z-y) > \frac{(n^*-3)\mathrm{OPT}}{n^*} - x$ (since $y> x+z$). The  weight of job $4$, which is at least the average among the remaining $n^*-3$ is more than $\frac{\mathrm{OPT}}{n^*} - \frac{x}{n^*-3}$. However, by the sorting, $w_4 \le w_3$ implying that  $\frac{x}{n^*-3} > y$, which contradicts $y>x+z$ and $z\ge0$. Thus,
$w_1+w_3 \ge \frac {2 \cdot \mathrm{OPT}}{n^*}$.
\end{proof}

Consider an NE profile $\sigma$. If job $1$ is covered, then, similar to the analysis of $n^*=3$, at least one of job $2, 3$ is covered, or has a beneficial migration, implying that either $\mathrm{val}(\sigma) \ge \frac{2 \cdot \mathrm{OPT}}{n^*}$ (by Claim~\ref{cl:J1andJ23}), or contradicting the stability of $\sigma$.

Therefore, we are left with the case that job $1$ is not covered.
If $\mathrm{val}(\sigma)<\frac{2\mathrm{OPT}}{n^*}$ it must be that job $1$ intersects with both $2, 3$. 
We first analyze the case $w_2+w_3 \ge \frac{2\mathrm{OPT}}{n^*}$: 
\begin{itemize}
    \item If both $2$ and $3$ are covered in $\sigma$, then $\mathrm{val}(\sigma)\ge\frac{2\mathrm{OPT}}{n^*}$.
    \item If one of $2$ and $3$ is covered, say $a$, then assume that job $a$ is assigned in interval $[t,t+p_a)$. 
Note that at least one of the intervals $[0,t)$, and $[t+p_a, T)$ has length at least $\frac{T-p_a}{2}$. Let $b=5-a$ be the other job in $\{2,3\}$. Since all the jobs in $A_3$ are covered in $\mathrm{OPT}$, we have that $p_1+p_b \le T-p_1$, thus $\min\{p_1,p_b\} \le \frac{T-p_1}{2}$, implying that the shortest unassigned job in $A_3$ has a beneficial migration if $\mathrm{val}(\sigma)<\frac{2\mathrm{OPT}}{n^*}$.
\item If both $2$ and $3$ are not covered, then the shortest of these two jobs has a beneficial migration by moving before or after job$1$if $\mathrm{val}(\sigma) <\frac{2\mathrm{OPT}}{n^*}$ (see the analysis for $n^*=3$).
    \end{itemize}
    
We turn to analyze the case $w_2+w_3 < \frac{2\mathrm{OPT}}{n^*}$. Consider the $4$th most profitable job in $\mathrm{OPT}$. Recall that $\sum_{i=1}^4 w_i \ge \frac {4 \cdot \mathrm{OPT}}{n^*}$. Since we analyze the case $w_1<\frac{2\mathrm{OPT}}{n^*}$, and $w_2+w_3 < \frac{2\mathrm{OPT}}{n^*}$, we have that $w_2+w_3+w_4 \ge \frac{2\mathrm{OPT}}{n^*}$ and $w_1+w_4 \ge  \frac{2\mathrm{OPT}}{n^*}$. 
That is, for all $k \in \{2,3,4\}$ we have that $w_1+w_k \ge \frac{2\mathrm{OPT}}{n^*}$. Consider the three jobs $\{2,3,4\}$.
\begin{itemize}
    \item If all the three are covered in $\sigma$ then $\mathrm{val}(\sigma) \ge w_2+w_3+w_4 \ge \frac{2\mathrm{OPT}}{n^*}$.
    \item If at most one of the three is covered, then the shortest uncovered job has a beneficial migration by moving before or after job$1$if $\mathrm{val}(\sigma) <\frac{2\mathrm{OPT}}{n^*}$ (see the analysis for $n^*=3$).
    \item If exactly two are are covered, then 
let $x$ and $y$ be covered, while $z$ is not. Assume the covered jobs are placed in $[s_x,f_x)$ and $[s_y,f_y)$, respectively, where w.l.o.g., that $f_x \le s_y$. 

If $\mathrm{val}(\sigma)<\frac{2\mathrm{OPT}}{n^*}$ and $\sigma$ is an NE, it must be that job $1$ cannot be placed before $y$, thus, $s_y<p_1$. In addition, $z$ cannot be placed after $y$, thus $T-f_y < p_z$.

Given that $\mathrm{OPT}$ processes all four jobs, we have that $p_1+p_x+p_y+p_z \le T$. However, the above inequalities imply that 
$p_1+p_y+p_z > s_y +(f_y-s_y)+T-f_y=T$,
contradicting the stability of $\sigma$ or the assumption that $\mathrm{val}(\sigma)<\frac{2\mathrm{OPT}}{n^*}$.
\end{itemize}

We conclude that if $n^*<n$, then $\mathrm{PoA}(G) \le \frac{n^*}2 \le \frac{n-1}{2}$. 
\end{proof}

Next, we show that the above analysis is tight.
\begin{theorem}
\label{thm:PoAtight}
For every $n \ge 2$, there exists a game $G_n \in \G_{single}$ with $n$ players such that $\mathrm{PoA}(G_n)=2$. For every $n \ge 5$, there exists a game $G_n$ with $n$ players such that $\mathrm{PoA}(G_n) \ge \frac{n-1}{2}$.
\end{theorem}

\begin{proof}
For $n \ge 2$, the game $G_n$ consists of $n$ unit weight jobs of length $1$. Let $T=2$. In a possible NE profile, all the jobs are placed in $[t,t+1)$ for some $0<t<1$, thus $\mathrm{val}(NE)=1$, while $\mathrm{OPT}=2$. 

For $n \ge 5$, the game $G_n$ is defined as follows: Let $T=n-1$. Player $1$ has one job with length $p_1=T$ and $w_1=2$. For $i=2,\ldots,n$,  $i$ has one job with $p_i=w_i=1$. In the Nash equilibrium, player $1$ places its job on $[0,T)$ and all players $i=2,\ldots, n$ place their jobs on $[0,1)$. Job $1$ will be covered and even if job $i\neq 1$ deviates to a different interval the machine will still process job $1$ as the machine breaks ties in favor of lexicographically small subsets. In a possible social optimum, jobs $i=2,\ldots,n$ each choose interval $[i-2,i-1)$ and are covered for a weight of $n-1$. Hence, $\mathrm{PoA}(G_n) \geq (n-1)/2$.
\end{proof}

\subsection{Weight Proportional to Length}
For the restricted class for games in $\G_{prop} \cap \G_{single}$, we show that the PoA is bounded by a constant. Recall that $G \in \G_{prop}$ if for all $i \in \J$, we have that $w_j=p_j$. Thus, the goal is to cover as much of the interval as possible in $[0,T)$.

\begin{theorem}\label{th:w=p}
For every $G \in \G_{prop} \cap \G_{single}$, $\mathrm{PoA}(G)\le 3$.
\end{theorem}
\begin{proof}

%
%
%
 Let $G \in \G_{prop}$ and let $\sigma$ be an NE profile of $G$. If all the jobs are covered, then clearly, $\mathrm{val}(\sigma)=\mathrm{OPT}$. Thus, we assume below that at least one job is not covered. During the interval $[0,T)$, the machine alternates between being idle and busy in $\sigma$. Let $u_1, b_1,u_2,b_2,\ldots,u_k,b_k,u_{k+1}$ denote the lengths of the intervals in which the machine is idle and busy. If two jobs are covered one after the other we consider them as two busy intervals with an idle interval of length $0$ between them. We show that the total idle time is strictly less than two times the total busy time.

Consider the first idle interval $u_1$, and any non-covered job, $i$. Assume that job $i$ is placed in $[0,p_i)$.
     If the point $p_i$ is within an idle interval, then $\sigma$ is not an NE, since it is profitable for the machine to cover $p_i$ instead of the jobs in $[0,p_i)$. Therefore, the point $p_i$ must be within a busy interval, say $b_a$ (See Fig.~\ref{fig:prop3}). Since $\sigma$ is an NE, placing job $i$ in $[0,p_i)$ is not beneficial. Specifically, if the machine covers job $i$, its added profit will be $p_i$ and its lost profit will be $\sum_{j=1}^a b_j$. Therefore, $p_i < \sum_{j=1}^a b_j$. The choice of $a$ is such that $p_i > \sum_{j=1}^{a-1} (u_j+b_j) +u_a$. Combining the two inequalities, we get that $b_a > \sum_{j=1}^{a} u_j$. 
    
    The same argument can be applied on the suffix of the schedule, starting from $u_{a+1}$. If $\sum_{j=1}^a (u_j+b_j)+p_i\leq T$, that is, job $i$ can fit after $b_a$, then we repeat the above procedure to the suffix of the schedule starting after $b_a$.  If $\sum_{j=1}^a (u_j+b_j)+p_i>T$, then we assume that job $i$ is placed in $[T-p_i,T)$. If point $T-p_i$ is within an idle interval, then $\sigma$ is not an NE. Therefore, the point $T-p_i$ must be in a busy interval, say $b_{a'}$, where $a'\leq a$ (because job $i$ does not fit after $b_a$). Using the same argument as above implies that $b_{a'}>\sum_{j=a'+1}^{k+1} u_j$.

    Adding the inequalities on the suffix and the prefix of $\sigma$  we get that $2\sum_{j=1}^k b_j > \sum_{j=1}^{k+1} u_j$.

\begin{figure}[ht]
\centering
\includegraphics[width=0.9\textwidth]{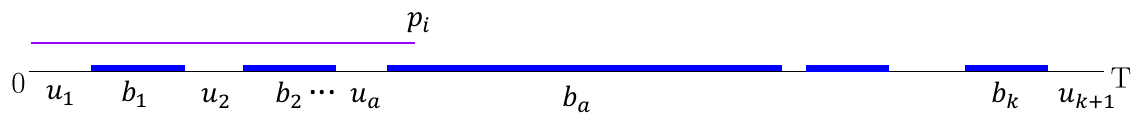}
\caption{\normalfont{
 The stability of $\sigma$ implies that point $p_i$ is within a busy interval.}}
\label{fig:prop3}
\end{figure}

The interval $T$ is partitioned such that $T = \sum_{j=1}^k b_j + \sum_{j=1}^{k+1} u_j$. Hence, $3\mathrm{val}(\sigma)=3\sum_{j=1}^k b_j >\sum_{j=1}^k b_j + \sum_{j=1}^{k+1} u_j=T\ge \mathrm{OPT}$, implying that $\mathrm{PoA}(G) < 3$.
\end{proof}

 \section{Arbitrary Number of Jobs per Color}

For games in $\G_{single}$, we have shown that an NE always exists, the price of stability is $1$ and the price of anarchy is $(n-1)/2=(c-1)/2$ for $n\geq 5$. In general, when several jobs may have the same color, the game becomes less stable, and the equilibrium inefficiency increases. 
Since most of the challenges arise already in games with two players, we first analyze this setting.

 \subsection{Two Players}

As shown in the introduction, there exists a game with $c=2$ that has no NE profile. We strengthen this negative result and show that NE is not guaranteed to exist even in the class $\G_{prop}$ where jobs' weights are proportional to their lengths.


\begin{proposition}
\label{ob:noNEprop}
There exists a game $G \in \G_{prop}$ with $c=2$ that has no NE profile.
\end{proposition} 

\begin{proof}
Consider a game in which $T=3$, the set $J_1$ includes jobs of lengths $p_1=3$ and  $p_2=1$, and the set $J_2$ includes $4$ jobs of length $1 - \epsilon$ (See Fig.~\ref{fig:noNEProp}).

\begin{figure}[ht]
\centering
\includegraphics[width=0.7\textwidth]{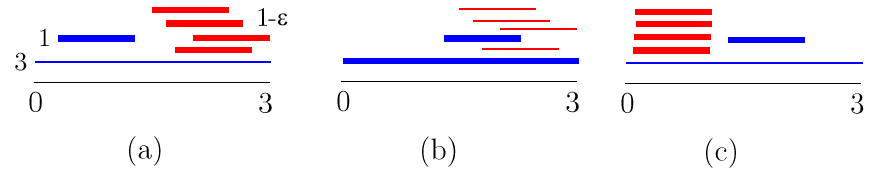}
\caption{\normalfont{
A game $G \in \G_{prop}$ with $c=2$ that has no NE.}}
\label{fig:noNEProp}
\end{figure}
If job $2$ does not overlap with any of the four jobs of $J_2$ then the machine will service all $5$ short jobs, and the players' utilities are $(1, 4(1-\epsilon))$. If job $2$ overlaps with at least one job of $J_2$ then the machine can cover at most $4$ short jobs, of total weight at most $4 -3\epsilon$, and will therefore service only jobs of $J_1$. The players' utilities are $(4,0)$.  Therefore, the game has no NE - player $1$ will always move job $2$ such that it overlaps a job from $J_2$, and player $2$ will move its jobs such that they do not overlap job $2$. Since the jobs of $J_2$ can overlap each other, player $2$ always has a valid move. 
\end{proof}

Next, we show that the problem of deciding whether a game has an NE, as well as the simpler problem of best-response computation are both NP-hard.

\begin{theorem}
\label{thm:SNPHdecide}
For $c=2$, deciding whether a game has an NE is strongly NP-hard.
\end{theorem}
\begin{proof}
We present a reduction from the $3$-Partition problem. The input of the $3$-Partition problem consists of a set $A$ of $3k$ integers $\{a_1, \dots, a_{3k}\}$ and an integer $B$, such that $\sum a_i = kB$ and for all $1 \le i \le 3k$, $B/4 < a_i < B/2$. The goal is to partition $A$ into $k$ triplets, each summing exactly to $B$. This problem is NP-hard in the strong sense. 

Given $A$, we construct a game with $c=2$. Let $T=B$. The set $J_1$, controlled by player 1, consists of $3k$ {\em partition jobs} with lengths $p_i = a_i$ and weights $w_i = 1$, as well as one job with length $p_{3k+1} = T$ and weight $w_{3k+1} = 1$. Player 2 controls a single job $j_0$ with $p_0<\frac{1}{4}$ and weight $w_{0} = k+\frac 1 2$.

Assume that a $3$-Partition of $A$ exists. Player $1$ can arrange the $3k$ jobs into $k$ triplets. Since each triplet sums to $T$, player $1$ can place these $k$ triplets such that the three jobs composing each triplet perfectly cover the interval $[0, T)$. Job $3k+1$ is placed in $[0,T)$. 
For every $t \in [0, T)$, the machine sees jobs of $J_1$ with a total weight of $k+ 1$. Independent of where player $2$ places $j_0$, the machine will choose to cover all the jobs of $J_1$. The players' utilities are $(3k+1, 0)$. Player $1$ is clearly stable, and player $2$ cannot place job $j_0$ such that it will be covered. Therefore, this is an NE.  

Assume that a partition of $A$ does not exist. We show that no NE exists, distinguishing between profiles in which job $j_0$ is covered or not. Let $\sigma$ be a profile in which job $j_0$ is not covered. Consider the assignment of the jobs of $J_1$. Clearly, if $j_0$ is not covered, then all the jobs of $J_1$ are covered and processed by the machine. 
\begin{claim}
If $A$ has no $3$-partition, then there must be an interval of length $1/4$ along which the machine processes, in addition to job ${3k+1}$, at most $k-1$ jobs of $J_1$.
\end{claim}
\begin{proof}
Assume towards contradiction that for every interval of length $1/4$ there are at least $k$ jobs of $J_1$ that intersect that interval. In particular, at least $k$ jobs intersect with the interval $[0,1/4]$, at least $k$ jobs intersect with the interval $[T/2-1/8,T/2+1/8]$, and at least $k$ jobs intersect with the interval $[T-1/4,T]$. Observe that a job that intersects $[0,1/4]$ cannot intersect $[T/2-1/8,T/2+1/8]$ or $[T-1/4,T]$, because $a_i\leq T/2-1$. This implies that we can partition $J_1$ into three sets of $k$ jobs: $J_{1,A}=\{j\in J_1\mid\sigma_i\cap[0,1/4]\}$, $J_{1,B}=\{j\in J_1\mid\sigma_i\cap[T/2-1/8,T/2+1/8]\}$, and $J_{1,C}=\{j\in J_1\mid\sigma_i\cap[T-1/4,T]\}$.

Let $j_a\in J_{1,A}$. Then there exists some job $j_b\in J_{1,B}$ that overlaps with $j_a$ or begins at most $1/4$ after $j_a$ ends. If not, there is an interval of size $1/4$ with less than $k$ jobs. Also, there exists some job $j_c\in J_{1,C}$ that overlaps with $j_b$ or begins at most $1/4$ after $j_b$ ends. Observe that this triplet sums to more than $T-1$. Removing these three jobs, we can repeat this argument for the remaining jobs implying that we must be able to partition the jobs into triplets each sums to more than $T-1$, so there must be a $3$-partition.
\end{proof}

Let $[t,t+\frac 1 4]$ be an interval along which the machine processes, in addition to job ${3k+1}$, at most $k-1$ jobs of $J_1$. Assume that player $2$ places job $j_0$ in $[t,t+\frac 1 4]$. The machine will benefit from covering  $j_0$ - increasing its profit by $k+\frac 1 2$ and giving away profit at most $k$. Therefore, $\sigma$ is not an NE.

Let $\sigma$  be a profile in which job $j_0$ is covered. Assume it is placed in interval $[t,t+\frac 1 4)$. It must be that job ${3k+1}$ is not covered. By placing at least $k$ jobs such that they overlap $[t,t+ \frac 1 4)$, Player $1$ can cause the machine to cover all the jobs of $J_1$. Therefore, $\sigma$ is not an NE. 
\end{proof}

\begin{theorem}
\label{thm:BRD2Shard}
For $c=2$, computing a best-response of a player is strongly NP-hard. 
\end{theorem}
\begin{proof}
We present a reduction from the $3$-Partition problem.  Given a set $A$ of $3k$ integers $\{a_1, \dots, a_{3k}\}$ and an integer $B$, such that $\sum a_i = kB$, and for all $1 \le i \le 3k$, $B/4 < a_i < B/2$, construct the following game with $c=2$. 
Let $T=2k-1$.
Player $1$ controls the set $J_1$ of $3k+1$ jobs, where for $0 \le i \le 3k$, $p_i=1$ and $w_i=a_i$. The last job in $J_1$ has length $T$ and weight $\frac 1 2$. 
Player $2$ controls $k$ unit-length jobs of weight $B$.

For every $1 \le \ell \le k$, let $[2(\ell-1), 2\ell-1)$ be denoted {\em the $\ell$-th odd interval}.
That is $[0,1)$ is the first odd interval, $[2,3)$ is the second odd interval, and so on, until $[2k-2, 2k-1)$ which is the $k$-th odd interval.
Let $\sigma$ be a profile in which, for $1 \le \ell \le k$, the $\ell$-th job of $J_2$ is placed in the $\ell$-th odd interval. We show that player $1$ has a best-response that leads to utility $kB+\frac 1 2$ if and only if a $3$-partition of $A$ exists.
Since job $3k+1$ has length $T$, it must be placed in $[0,T)$. 
Since the partition jobs have length $1$, each of them can overlap at most one odd interval. 

If a $3$-partition of $A$ exists, then player $1$ can place a triplet of jobs of total weight $B$ in every odd interval. In this schedule, there are jobs of total weight $B$ from each color on each odd interval. The long job of $J_1$ will cause the machine to process all the jobs of $J_1$, for a total revenue of $kB+\frac 1 2$. This is clearly a best-response for $J_1$.

If a partition of $A$ does not exist, then the total weight of jobs from $J_1$ overlapping at least one odd interval is at most $B-\frac 1 2$. To see this, suppose, by way of contradiction, that the total weight of jobs from $J_1$ overlapping each odd interval is more than $B-\frac 1 2$, that is, at least $B$. This implies that each odd interval is overlapped by at least and thus exactly three partition-jobs, but this would imply the existence of a $3$-partition. Therefore, the machine will process at least one of the jobs of $J_2$ and will not process job $3k+1$ of $J_1$.
Thus, the maximum utility player $1$ can gain if a partition does not exist is $kB$. 
\end{proof}



Finally, we provide tight analysis of the equilibrium inefficiency in a $2$-player game.  
\begin{theorem} 
\label{thm:C2PoS2}
For $c=2$, $\mathrm{PoA}(G)\leq 2$. Moreover, there exists a game $G$ with $c=2$ and $\mathrm{PoA}(G)=2$, and for every $\epsilon>0$, there exists a game $G'$ with $c=2$ and $\mathrm{PoS}(G')=2-\epsilon$. 
\end{theorem} 
\begin{proof}
Let $\sigma$ be an NE. It must be that $\mathrm{val}(\sigma) \ge \frac 1 2 (w(J_1)+w(J_2))$, because the machine can always fully cover $J_1$ or $J_2$ gaining profit at least $\max(w(J_1),w(J_2))$. Since $\mathrm{val}(\mathrm{OPT}) \le w(J_1)+w(J_2)$, we have that $\mathrm{PoA}(G) \le 2$. Also, $\mathrm{PoA}(G)=2$ is achieved in a game $G$ in which $T=2$, and each set includes a single unit-weight job of length $1$. If both players locate their job in $[t, t+1)$ for some $0<t<1$, we get an NE of profit $1$. In a socially optimum solution, the jobs are assigned in $[0,1)$ and $[1,2)$ for a total profit of $2$.


Given $0<\epsilon<1$, select $\delta$ such that $\epsilon=\frac{2\delta}{1+\delta}$. Consider a game $G'$ in which $T=2$, the set $J_1$ includes two jobs where $p_1=2, w_1=\delta$ and $p_2=1, w_2=1$, and $J_2$ includes a single job where $p_3=1$ and $w_3=1$.

\begin{figure}[ht]
\centering
\includegraphics[width=0.4\textwidth]{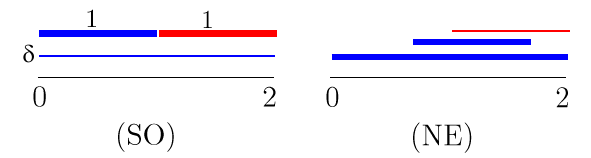}
\caption{\normalfont{
A game with $c=2$ and PoS$= \frac {2}{1+\delta} =2-\epsilon$.}}
\label{fig:PoSis2}
\end{figure}
Fig.~\ref{fig:PoSis2} presents a socially optimum profile and an NE profile. In the social optimum, job $2$ is covered in $[0, 1)$, and job $3$ is covered in $[1, 2)$.  The players utilities are $(1,1)$. This schedule is not an NE, since player $1$ can move job $2$ to $(t,t+1)$ for any $0<t<1$. The machine will then processes the two jobs of $J_1$ and the players' utilities will be $(1+\delta, 0)$. The only NE profiles of $G'$ are those in which player $1$ locates job $1$ in $(t,t+1)$ for any $0<t<1$. The profit of such a schedule is $1+\delta$, while $\mathrm{val}(\mathrm{OPT})=2$. Thus, $\mathrm{PoS}(G') = \frac {2}{1+\delta}= 2 -\epsilon$. 
\end{proof}


\subsection{Equilibrium Inefficiency - Arbitrary Number of Players}

We show that both the price of anarchy and the price of stability are linear in the number of colors.

\begin{theorem}\label{thm:poapos}
For all $G$, $\mathrm{PoA}(G)\leq c$ and for every $c>2$ and $\epsilon'>0$, there exists a game $G$ with $\mathrm{PoS}(G)\geq \frac{c}{2}-\epsilon'$.
\end{theorem}


\begin{proof}
 Let $\sigma$ be an NE. Since the machine can always fully cover all jobs of the most profitable color, $\mathrm{val}(\sigma) \ge \frac 1 c \sum_{i=1}^c w(J_i)$. Since $\mathrm{val}(\mathrm{OPT}) \le \sum_{i=1}^c w(J_i)$, we have that $\mathrm{PoA}(G) \le c$.

For the lower bound on the price of stability, Given $c>2$ and $\epsilon'>0$, let $T=2c+1$ and let $\epsilon$ be a small number such that $\frac{c+\frac{c+1}{c+3}\epsilon}{2+\epsilon} = \frac c 2 -\epsilon'$. The game $G$ is defined as follows. Every player $i=1,\ldots,c$ has one job with $p_i=w_i=1$. Player $1$ has, in addition, one job with length $p_{1'}=T$ and weight $w_{1'}=1+\epsilon$. 
Player $2$ has, in addition, $c+2$ light jobs of lengths $2,\ldots,c+3$ and weight $\frac{\epsilon}{c+3}$ each.
In a possible NE profile, $s_0$ (see Fig.~\ref{fig:pos}(a), where $c=4$), player $1$ puts its long job in $[0,T)$ and all players $i=1,\ldots, c$ put their unit job in $[0,1)$. 
Player $2$ puts each job of length $t$ in $[0, t)$. 
The two jobs of player $1$ will be covered, for a total profit of $2+\epsilon$. Player $1$ is clearly stable. If player $i >2$ deviates to a different interval the machine will still process the two jobs of player $1$ as $2 + \epsilon \geq 2 + \frac{c + 2}{c+3} \epsilon$. 
Since the total weight of the light jobs of player $2$ is less than $\epsilon$, player $2$ does not have a beneficial migration as well.

\begin{figure}[ht]
\centering
\includegraphics[width=\textwidth]{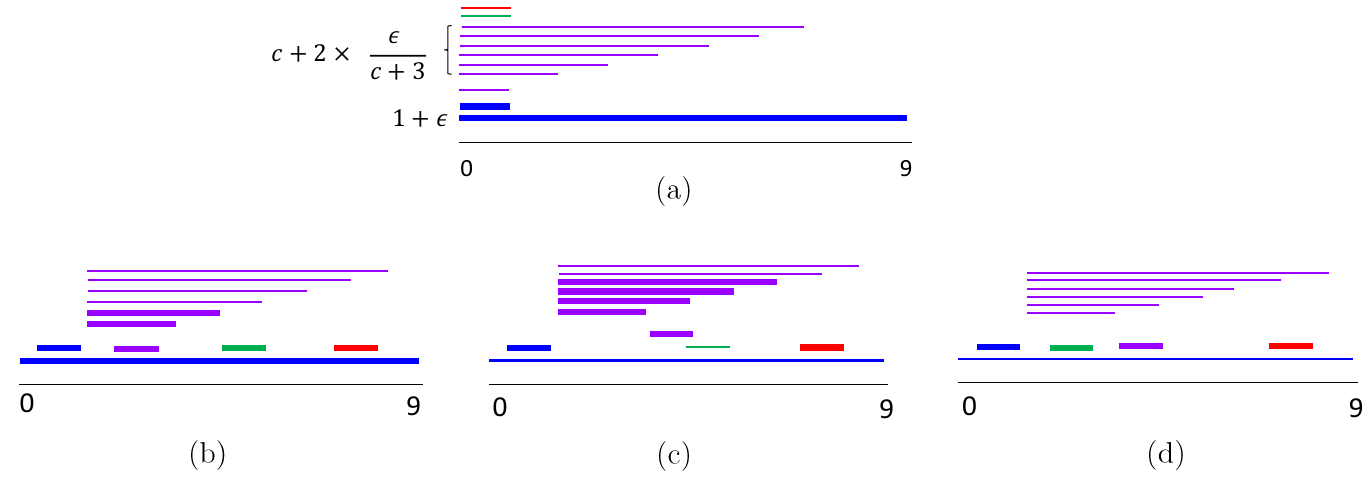}
\caption{\normalfont{
A game with $c=4$ and PoS$=2-\epsilon'$. Bold intervals are covered.
(a) an NE profile. (b) all unit jobs and some jobs of $J_2$ are covered. (c) a profile resulting from a beneficial deviation of player $2$. (d) a best-response of player $3$; once again, all unit jobs are covered.}}
\label{fig:pos}
\end{figure}

We show that every NE has profit $2+\epsilon$. Consider a profile $\sigma$. If the long job of $J_1$ is covered, then only jobs of $J_1$ are covered and the profit is $2+\epsilon$. If the long job of $J_1$ is not covered, then we show that $\sigma$ is not an NE. Since $T=2c+1$, all the unit-jobs are covered. Indeed, a player that controls a non-covered unit-job can place it on an idle unit slot or overlap (and be preferred) over light jobs of $J_2$. Let $x>0$ be the length of the longest interval in $\sigma$ that has no unit job of $J_j$, $j \neq 2$. 
If all the unit jobs are covered, then $c-1$ out of the $T$ available slots are busy processing unit jobs not in $J_2$, implying that $1 \le x \le T-(c-1) = c+2$. 

Let $j$ be a player whose unit-job is covered adjacent to the interval of length $x$. 
Player $2$ can deviate such that its unit job overlaps $j$ and all the light jobs of length at most $x+1$ are covered (see Fig.~\ref{fig:pos}(c)). Thus, $\sigma$ is not an NE.
The resulting schedule is not an NE as well, since job $j$ is non-covered, and as mentioned above, if the long job of $J_1$ is not covered, then all unit-jobs are.

In a social optimum solution, all the jobs except for the long job of $J_1$ and the longest light job of $J_2$ are covered. Hence, $\mathrm{PoS}(G) = \frac{c+\frac{c+1}{c+3}\epsilon}{2+\epsilon} = \frac c 2 -\epsilon'$.  \end{proof}



\subsection{Unit-length Jobs} 

In this section we analyze games with unit-length jobs. For this class, we provide positive results. Specifically, some optimal solution is an NE and the PoA is bounded by a constant less than $3$. 
\begin{theorem}
\label{thm:unitPoS}
For every game $G \in \G_{unit}$, $G$ has an NE profile and $\mathrm{PoS}(G)=1$.  
\end{theorem}
\begin{proof}
Denote by $W_i$ the total weight of jobs in $J_i$, and assume w.l.o.g., that $W_1 \ge \dots \ge W_c$. Let $k=\floor{T}$. 
In a social optimum solution, $\sigma^*$, the $k$ most profitable colors are serviced, where w.l.o.g., all jobs of $J_i$, for $1 \le i \le k$ are serviced in $[i-1, i)$. 
The players' utility vector is $(W_1,\ldots,W_k,0,\ldots,0)$.
The total revenue is $\mathrm{val}(\sigma^*)=\sum_{i=1}^k W_i$. This schedule is an NE, since the total weight controlled by any non-covered player is at most $W_k$. Therefore, $\mathrm{PoS}(G)=1$.
\end{proof}



\begin{theorem}
\label{thm:unitPoA}
 Let $G \in \G_{unit}$ and let $k=\lfloor T \rfloor$. We have $\mathrm{PoA}(G)  \le \min\{3 - \frac{2}{k}, 3-\frac 2 c\}$. In addition, for every even $c$, there exists a game $G \in \G_{unit}$ with $\mathrm{PoA}(G)  = 3 - \frac{2}{k}= 3 -\frac 2 c$.
\end{theorem}
\begin{proof}
Let $G \in \G_{unit}$. We show that $\mathrm{PoA}(G)  \le 3 - \frac{2}{k}$. Consider an NE schedule $\sigma$ of $G$. Clearly, if some job of color $i$ is covered in $\sigma$, then all jobs of color $i$ are covered, since uncovered jobs can be added to fully overlap a covered one.

Recall that the sets are sorted in non-increasing order of total weight. Let $\ell>1$ be the lowest index of a non-covered color. 
For every $0 \le t \le T-2$, jobs of total weight at least $W_{\ell}$ start processing in the interval $[t,t+2)$, as otherwise, by assigning all the jobs of $J_{\ell}$ in $[t+1,t+2)$, it would be beneficial for the machine to cover the jobs of $J_{\ell}$. 

If $k$ is even, then consider the partition of $[0,T)$ into $\frac k 2$ intervals $\{[0,2), [2,4), \ldots [k-2,T)\}$. The first $\frac k 2-1$ intervals are of length $2$, and the interval $[k-2,T)$ is of length at least $2$. We conclude that $\mathrm{val}(\sigma) \ge \frac {k } 2\cdot W_{\ell}$, that is, $W_{\ell}\le \frac {2 \cdot \mathrm{val}(\sigma)}{k}$.
If $k$ is odd, then consider the partition of $[0,T)$ into $\ceil{\frac k 2}$ intervals $\{[0,2), [2,4), \ldots [k-1,T)\}$. The first $\floor{\frac k 2}$ intervals are of length $2$ and the interval $[k-1,T)$ is of length at least $1$. No job can start its processing after time $T-1$ and complete on time. Thus, jobs of total weight at least $W_{\ell}$ start their processing during $[k-1, T-1)$, as otherwise it would be beneficial for the machine to cover the jobs of $J_{\ell}$ in the interval $[T-1,T)$. We conclude that $\mathrm{val}(\sigma) \ge \ceil{\frac{k}{2}}\cdot W_{\ell} \ge \frac{k}{2}\cdot W_{\ell}$, that is, $W_{\ell}\le \frac {2 \cdot \mathrm{val}(\sigma)}{k}$.

By the choice of $\ell$, $\mathrm{val}(\sigma) = \sum_{i=1}^{\ell-1} W_i$. 
In addition $\mathrm{val}(\sigma^*) = \sum_{i=1}^{k} W_i \le \sum_{i=1}^{\ell-1} W_i + (k-\ell+1)W_{\ell}$.
Since $\ell >1$, we have $\mathrm{val}(\sigma^*) \le \mathrm{val}(\sigma)+ \frac{2(k-1)\mathrm{val}(\sigma)}{k} = (3-\frac{2}{k})\mathrm{val}(\sigma)$. That is, $\mathrm{PoA}(G) \le 3-\frac{2}{k}$.

If $k \le c$, then $\min\{3 - \frac{2}{k}, 3-\frac 2 c\} = 3 - \frac{2}{k}$, and we are done. If $k>c$ then note that $W_{\ell}\le \frac {2 \cdot \mathrm{val}(\sigma)}{c}$. Also, $\mathrm{val}(\sigma^*) = \sum_{i=1}^{c} W_i \le \sum_{i=1}^{\ell-1} W_i + (c-\ell+1)W_{\ell}$.
Since $\ell >1$, we have $\mathrm{val}(\sigma^*) \le \mathrm{val}(\sigma)+ \frac{2(c-1)\mathrm{val}(\sigma)}{c} = (3-\frac{2}{c})\mathrm{val}(\sigma)$. That is, $\mathrm{PoA}(G) \le 3-\frac{2}{c}$.

For the lower bound, given an even number $c$, consider a game $G$ in which $c=T=k$. The set $J_1$ consists of $\frac{k}{2}$ unit-weight jobs, and for all $2 \le i \le k$, the set $J_i$ consists of a single unit-weight job.
For $1 \le j \le \frac k 2$, let $a_j=j(2-\epsilon)$. Consider a schedule in which for all $1 \le j \le \frac k 2$ there is one job of $J_1$ in the slot $[a_j-1, a_j)$, and all other jobs are in $[1-\epsilon, 2-\epsilon)$. Note that the idle intervals between the jobs have length $1-\epsilon$, and the last idle interval $[a_{k/2},T)$ has length $k\epsilon/2$. Therefore, a schedule in which only the jobs of $J_1$ are covered is an NE (we assume that the tie-breaking of the machine in $[1-\epsilon, 2-\epsilon)$ is in favor of  $1$. This assumption can be removed by increasing to $(1+\epsilon)$ the weight of the job of $J_1$ placed in this slot). The players' utility vector is $(\frac k 2,0,\ldots,0)$. The total revenue is $\frac k 2$. The social optimum for this instance has profit $\sum_{i=1}^k W_i = \frac{k}{2}+k-1 = \frac{3k}{2}-1$.  We conclude that $\mathrm{PoA}(G)=3-\frac{2}{k}$. 
\end{proof}

\section{Extension: Jobs with Release times and Due-dates} 

A natural extension to our game considers a setting in which every job $j \in \J$ is associated also with a release time, $r_j$, and a due-date, $d_j$. The player  controlling $J_i$ should place every job colored $i$ in an interval $[t, t+p_j) \subseteq I_j=[r_j,d_j)$. 
Denote by $\G_{real-time}$ the corresponding class of games.
Note that from the machine's point of view, the problem remains the same. However, for the players, this setting is computationally harder, and even a $2$-player game with unit-length unit-weight jobs, may not have an NE. 



\begin{theorem}
\label{thm:noNEnonsymm}
    There exists a game $G \in \G_{real-time} \cap {\G}_{unit}$ with $c=2$ that has no NE.
\end{theorem}

\begin{proof}
Let $G$ be the following game with $T=3$, $|J_1|=7$, and $|J_2|=5$. The set $J_1$ consists of seven unit-length unit-weight jobs. The first two are restricted to go to $[0,1)$, the next two are restricted to go to $[2,3)$, and the remaining three can be placed anywhere in $[0,3)$. The set $J_2$ consists of five unit-length unit-weight jobs. The first is restricted to go to $[0,1)$, the second is restricted to go to  $[2,3)$, and the remaining three can be placed anywhere in $[0,3)$. Consider a profile $\sigma$. 
For $i=1,2$, denote by $x_i$ and $y_i$ the number of jobs of $J_i$ that intersects with 
$[0,1)$ and $[2,3)$, respectively.
Note that a unit-job cannot intersect with both $[0,1)$ and $[2,3)$, thus $x_1+y_1 \le 7$ and $x_2+y_2 \le 5$. In addition, the restricted jobs imply that $x_1 \ge 2, y_1 \ge 2, x_2 \ge 1$ and $y_2 \ge 1$. 

Assume that the utility of player $2$ in $\sigma$ is less than $4$. Since $\min\{x_1,y_1\} \le 3$, player $1$ has at most $3$ jobs intersecting at least one of $[0,1)$ or $[2,3)$. Player $2$ has a beneficial migration - by placing all its flexible jobs in the corresponding interval, it increases its utility to $4$.  Thus, $\sigma$ is not an NE. 

Assume now that the utility of player $2$ in $\sigma$ is at least $4$. Since only three jobs of $J_2$ can be covered in $[1,2)$, the machine must processes $J_2$ also in an interval $I$ that contains $[0,1)$ or $[2,3)$. Now, either  $2$ has a beneficial migration - by moving all its flexible jobs to $I$, or, if all the flexible jobs of $J_2$ are already in $I$, player $1$ can increase its utility to $7$ by placing all its flexible jobs in $I$. Thus, $\sigma$ is not an NE.
\end{proof}



\begin{theorem}
\label{thm:brdNonsymm}
Computing a best-response of a player is strongly NP-hard even for $G \in \G_{real-time} \cap {\G}_{unit}$ with $c=2$.
\end{theorem}
\begin{proof}
We present a reduction from the $3$-Partition problem.  Given a set $A$ of $3k$ integers $\{a_1, \dots, a_{3k}\}$ and an integer $B$, such that $\sum a_i = kB$, and for all $1 \le i \le 3k$, $B/4 < a_i < B/2$, construct the following game with $c=2$. Let $T=2k-1$. 
For every $1 \le \ell \le k$, let $[2(\ell-1), 2\ell-1)$ be denoted {\em the $\ell$-th odd interval}.
That is $[0,1)$ is the first odd interval, $[2,3)$ is the second odd interval, and so on, until $[2k-2, 2k-1)$ which is the $k$-th odd interval.

Player $1$ controls a set $J_1$ of $4k$ jobs. For $0 \le i \le 3k$, $p_i=1$ and $w_i=a_i$. For each of these jobs, $r_i=0$ and $d_i=2k-1$, that is, they can be processed anywhere in $[0,T)$. For $3k+1 \le i \le 4k$, let $\ell_i=i-3k$. Job $3k+\ell_i$ has $p_i=1, w_i=1/2$ and is restricted to be processed exactly in the $\ell$-th odd interval.

Player $2$ controls $k$ unit-length jobs of weight $B$.
For $1 \le \ell \le k$, job $\ell$ in $J_2$ is restricted to be processed exactly in the $\ell$-th odd interval.

We argue that player $1$ has a best-response that leads to utility $kB+k/2$ if and only if a $3$-partition of $A$ exists.

Observe that a utility of $kB+k/2$ for player $1$ implies that all jobs of player $1$ are covered. Since jobs $3k+1,\ldots,4k$ in $J_1$ are restricted to go to the same slots as the jobs of $J_2$, and since player $1$ cannot place the partition jobs such that they overlap more than a single odd interval, the machine will process all the jobs of $J_1$ only if player $1$ places them such that the load of its jobs overlapping each of  the odd intervals is more than $B$. Since we need exactly thee jobs on each odd interval to induce a load of more than $B$, this implies the existence of a $3$-partition.  
\end{proof}

\section{Conclusion and Directions for Future Work}
In this paper, we introduced and analyzed interval scheduling games with color-based concurrent job, modeling a strategic variant of the classical interval scheduling problem where each player controls a class of same-colored jobs. Driven by applications like wireless beamforming, players strategically place their jobs to maximize their own covered weight, while the machine dynamically optimizes global coverage.

Our work shows a significant difference between single-job  and multi-job classes. We proved that every single-job game in has a pure Nash equilibrium, the Price of Stability is $1$, and best-response dynamics always converge. We provided a polynomial-time algorithm to compute an equilibrium and established a tight linear bound of $(n-1)/2$ on the Price of Anarchy, which drops to a constant bound of $3$ under weight-proportional lengths.  On the other hand, when players control classes of multiple jobs, stability collapses. Pure Nash equilibria may not exist even under weight-proportional lengths, and both best-response computation and deciding equilibrium existence become strongly NP-hard. 
For games with unit-length jobs we recovered positive results, showing equilibrium existence, $PoS=1$, and a tight constant price of anarchy strictly less than $3$. However, adding release times and due dates disrupts this setting, making the game unstable and best-response computation NP-hard.  

Our work leaves open several avenues for future research. Some open questions follow directly from our work. For example, generalizing the classes of games for which an NE is guaranteed to exist, the tightness of the result in Theorem \ref{th:w=p}, and extending the results for games with release times and due-dates.

For single-job games ($\G_{single}$), an optimal Knapsack solution directly induces an NE. A natural question is whether standard approximation algorithms for Knapsack can be used to compute highly efficient states in polynomial time. While scheduling an approximate packing solution does not inherently yield a stable profile, future work could explore whether a stability can be restored while preserving strong social welfare guarantees.


A different direction is to study other machine environments. Our current work assumes that there is only one machine, what can we say if there are multiple (parallel) machines? Or what if the machine cannot schedule all jobs from a particular color simultaneously, but has a fixed capacity or incurs delays if multiple jobs are processed simultaneously?

\bibliographystyle{plain}


\end{document}